\documentclass[12pt,draftcls,onecolumn]{IEEEtran}

\IEEEoverridecommandlockouts

\usepackage{latexsym}
\usepackage{graphicx}
\usepackage{epsfig}
\usepackage{subfigure}
\usepackage{array}
\usepackage{amsmath}
\usepackage{amssymb}
\usepackage{color}
\usepackage{amsthm}
\usepackage{enumerate}
\theoremstyle{plain} 
\newtheorem{thm}{Theorem}
\newtheorem{cor}{Corollary}
\newtheorem{lemma}{Lemma}
\newtheorem{prop}{Proposition}
\theoremstyle{definition}

\newtheorem{remark}{Remark}

\DeclareMathOperator{\tr}{tr}

\newcommand{\bxi} {\boldsymbol{\xi}}
\newcommand{\bW}{{\bf W}}
\newcommand{\bQ}{{\bf Q}}

\newcommand{\bH}{{\bf H}}
\newcommand{\bR}{{\bf R}}
\newcommand{\bI}{{\bf I}}
\newcommand{\bU}{{\bf U}}
\newcommand{\bV}{{\bf V}}
\newcommand{\bM}{{\bf M}}
\newcommand{\bD}{{\bf D}}
\newcommand{\bP}{{\bf P}}
\newcommand{\bA}{{\bf A}}
\newcommand{\bB}{{\bf B}}
\newcommand{\bSigma} {\boldsymbol{\Sigma}}
\newcommand{\bu}{{\bf u}}
\newcommand{\bx}{{\bf x}}
\newcommand{\bLam}{\boldsymbol{\Lambda}}
\newcommand{\blam}{\boldsymbol{\lambda}}
\newcommand{\lam}{\lambda}
\def\bal#1\eal{\begin{align}#1\end{align}}
\newcommand{\bp} {\begin{proof}}
\newcommand{\ep} {\end{proof}}
\newcommand{\Lb} {\left(}
\newcommand{{\Rb}} {\right)}

\newcommand{{\Rf}} {\right\}}
\newcommand{\sN}{\mathcal{N}}
\newcommand{\sR}{\mathcal{R}}
\newcommand{\bi}{\begin{itemize}}
\newcommand{\ei}{\end{itemize}}

\begin{document}

\title{Rank-Deficient Solutions for Optimal Signaling over Wiretap MIMO Channels}

\author{Sergey Loyka, Charalambos D. Charalambous
\vspace*{-1.5\baselineskip}

\thanks{This paper was presented in part at IEEE Int. Symp. Inf. Theory (ISIT-13), Istanbul, Turkey, Jul. 2013, and at ISIT-14, Honolulu, HI, USA, Jun. 2014}

\thanks{S. Loyka is with the School of Electrical Engineering and Computer Science, University of Ottawa, Ontario,
Canada, K1N 6N5, e-mail: sergey.loyka@uottawa.ca}

\thanks{C.D. Charalambous is with the ECE Department, University of Cyprus, 75 Kallipoleos Avenue, P.O. Box 20537, Nicosia, 1678, Cyprus, e-mail: chadcha@ucy.ac.cy}
}
%
\maketitle
%

\begin{abstract}
Capacity-achieving signaling strategies for the Gaussian wiretap MIMO channel are investigated without the degradedness assumption. In addition to known solutions, a number of new rank-deficient solutions for the optimal transmit covariance matrix are obtained. The case of a weak eavesdropper is considered in detail and the optimal covariance is established in an explicit, closed form with no extra assumptions. This provides lower and upper bounds to the secrecy capacity in the general case with a bounded gap, which are tight for a weak eavesdropper or/and low SNR. Closed form solutions are also obtained for isotropic and omnidirectional eavesdroppers, based on which lower and upper bounds to the secrecy capacity are established in the general case.
Sufficient and necessary conditions for optimality of 3 popular transmission techniques, namely the zero-forcing (ZF), the standard water-filling (WF) over the channel eigenmodes and the isotropic signaling (IS), are established for the MIMO wiretap channel. These solutions are appealing due to their lower complexity. In particular, no wiretap codes are needed for the ZF transmission, and no precoding or feedback is needed for the isotropic signaling.
\end{abstract}


%
\section{Introduction}
\label{sec:introduction}

Widespread use of wireless systems on one hand and their broadcast nature on the other have initiated significant interest in their security. Information-theoretic studies of the secrecy aspects of wireless systems have recently attracted significant interest \cite{Bloch}. Due to the high spectral efficiency of wireless MIMO systems and their wide adoption by the academia and industry, the Gaussian MIMO wire-tap channel (WTC) has emerged as a popular model and a number of results have been obtained for this model, including the proof of optimality of the Gaussian signaling \cite{Bloch}-\cite{Oggier}.

An optimal transmit covariance matrix under the total power constraint has been obtained for some special cases (low/high SNR, MISO channels, full-rank or rank-1 solutions) \cite{Khisti-1}-\cite{Li-09}, but the general case remains elusive. The main difficulty lies in the fact that, unlike the regular MIMO channel, the underlying optimization problem for the MIMO-WTC is generally not convex. It was conjectured in \cite{Oggier} and proved in \cite{Khisti-2} using an indirect approach (via a degraded channel) that the optimal signaling is on the positive directions of the difference channel.  A direct proof (based on the necessary Karush-Kuhn-Tucker (KKT) optimality conditions) has been obtained in \cite{Loyka}, while the optimality of signaling on non-negative directions has been established in \cite{Li-09} via an indirect approach. Closed form solutions for MISO and rank-1 MIMO channels have been obtained in \cite{Khisti-1}\cite{Loyka}-\cite{Li-10a}. The 2-2-1 channel (2 transmit, 2 receive, 1 eavesdropper antenna) has been studied earlier in \cite{Shafiee-09}. The low-SNR regime has been studied in detail in \cite{Gursoy}. An exact full-rank solution for the optimal covariance and several of its properties have been obtained in \cite{Loyka}. In particular, unlike the regular channel (no eavesdropper), the optimal power allocation does not converge to the uniform one at high SNR and the latter remains sub-optimal at any finite SNR. In the case of a weak eavesdropper, the optimal signaling mimics the conventional one (water-filling over the channel eigenmodes) with an adjustment for the eavesdropper channel.

Finally, while no analytical solution for the optimal covariance is known in the general case, numerical algorithms have been developed to attack the problem in \cite{Li-13}-\cite{Alvarado}, which however suffer from the lack of provable global convergence due to the non-convex nature of the optimization problem in the general case. A globally-convergent numerical algorithm for the general case, which is based on an equivalent min-max reformulation of the original problem, was proposed in  \cite{Loyka-15} and its convergence was proved, which takes only a moderate or small number of steps in practice.

The present paper extends the known analytical results for the optimal covariance in several directions. First, motivated by a scenario where the legitimate receiver (Rx) is closer to the transmitter (Tx) than the eavesdropper, the case of a weak eavesdropper is studied and its optimal covariance is obtained in an explicit closed form without any extra assumptions in Section \ref{sec:Weak Eavesdropper}. It provides novel lower and upper bounds to the secrecy capacity in the general case with a bounded gap, which are tight when the eavesdropper is weak or/and the SNR is low and hence serve as an approximation to the true capacity. It also captures the capacity saturation effect at high SNR observed in \cite{Khisti-2}\cite{Loyka}.  The range of validity of this model is indicated.

The presence of the eavesdropper channel state information (CSI) at the transmitter is in question when the eavesdropper does not cooperate (e.g. to hide its presence). To address this issue, we consider in Section \ref{sec:isotropic} an isotropic eavesdropper model, whereby the Tx does not know the directional properties of the eavesdropper and hence assumes it is isotropic, i.e. the eavesdropper channel gain is the same in all directions. The secrecy capacity as well as an optimal signaling to achieve it and its properties are established in an explicit closed form. This case is shown to be the worst-case MIMO wire-tap channel. Based on this, lower and upper capacity bounds are obtained for the general case, which are achievable by the isotropic eavesdropper. The properties of the optimal power allocation are pointed out.

The case of isotropic eavesdropper above requires the number of its antennas to be not less than the number of Tx antennas (which is necessary for a full-rank eavesdropper channel), which may not be the case in practice. To address this issue, Section \ref{sec:omnidirectional} studies an omnidirectional eavesdropper, which may have a smaller number of antennas (and hence rank-deficient channel) and which has the same gain in any direction of a given subspace. The secrecy capacity and the optimal signaling are established in a closed form.

The case of identical right singular vectors of the Rx and eavesdropper channels is investigated and the optimal covariance is established in a closed from in Section \ref{sec:Identical Right Singular Vectors}. This case is motivated by a scenario where the legitimate receiver and the eavesdropper are spatially separated so that each has its own set of local scatterers inducing its own left singular vectors (SV), while both channel are subject to the same set of scatterers around the transmitter (e.g. a base station) and hence the same right SVs. This is similar to the popular Kronecker MIMO channel correlation model, see e.g. \cite{Kermoal}, where the overall channel correlation is a product of the independent Tx and Rx parts, which are induced by the respective sets of scatterers.

In Section \ref{sec:ZF}, the conditions for optimality of popular zero-forcing (ZF) signaling are established, whereby the Tx antenna array forms a null in the eavesdropper direction. Under those conditions, the standard eigenmode signaling and the water-filling (WF) power allocation on what remains of the required channel (after the ZF) are optimal. Furthermore, no wiretap codes are required as regular coding on the required channel suffices, so that the secrecy requirement imposes no extra complexity penalty (beyond the standard ZF). In this case, the optimal secure signaling is decomposed into two parts: part 1 is the ZF (null forming in the terminology of antenna array literature \cite{VanTrees}), which ensures the secrecy requirement, and part 2 is the standard signaling (eigenmode transmission, WF power allocation and coding) on the required channel, which maximizes the rate of required transmission. This is reminiscent of the classical source-channel coding separation \cite{Cover book}.

In Sections \ref{sec:WF} and \ref{sec:isotropic R}, we consider two other popular signaling techniques: the standard water-filling over the eigenmodes of the legitimate channel and the isotropic signaling (IS, whereby the covariance matrix is a scaled identity) and establish sufficient and necessary conditions under which they are optimal for the MIMO WTC. These techniques are also appealing due to a number of reasons. While the standard WF does require wiretap codes, standard solutions can be used for power allocation and eigenmode transmission (i.e. spatial modulation); the isotropic signaling is appealing due to its low complexity: no eavesdropper CSI is required at the transmitter as independent, identically distributed data streams are launched by each antenna. The set of channels for which the isotropic signaling is optimal is fully characterized in Section \ref{sec:isotropic R}. It turns out to be much richer than that of the conventional (no eavesdropper) MIMO channel.

\emph{Notations}: Lower case bold letters denote vectors while bold capitals denote matrices. $\lambda_i(\bW)$ denotes the eigenvalues of a matrix $\bW$ in decreasing order unless indicated otherwise; $(x)_+ = \max\{x,0\}$ for a scalar $x$; $\sN(\bW)$ and $\sR(\bW)$ are the null space and the range of a matrix $\bW$; $(\bW)_+$ denotes the positive eigenmodes of a Hermitian matrix $\bW$:
\bal
(\bW)_+ = \sum_{i: \lambda_i(\bW) >0} \lambda_i \bu_i \bu_i^{\dag}
\eal
where $\bu_i$ is $i$-th eigenvector of $\bW$; $\tr\bW$ and $|\bW|$ denote the trace and the determinant of $\bW$; $\bW^{\dag}$ is the Hermitian conjugation of $\bW$.

\vspace*{-0.5\baselineskip}
\section{Wire-Tap Gaussian MIMO Channel Model}
Let us consider the standard wire-tap Gaussian MIMO channel model,
\begin{equation}
\label{eq1}
{\rm {\bf y}}_1 ={\rm {\bf H}}_1 {\rm {\bf x}}+\bxi_1 , \quad {\rm {\bf y}}_2 ={\rm {\bf H}}_2 {\rm {\bf x}}+ \bxi_2
\end{equation}
where ${\rm {\bf x}}=[x_1 ,x_2 ,...x_m ]^T\in \mathbb{C}^{m,1}$ is the transmitted complex-valued signal vector of dimension $m \times 1$, ``T'' denotes transposition, ${\rm {\bf y}}_k \in \mathbb{C}^{n_k}$, $k=1,2$, are the received vectors at the receiver and eavesdropper, $\bxi_1$ and $\bxi_2$ are the circularly-symmetric additive white Gaussian noise at the receiver and eavesdropper (normalized to unit variance in each dimension), ${\rm {\bf H}}_k \in \mathbb{C}^{{n_k},m}$ is the $n_k \times m$ matrix of the complex channel gains between each Tx and each receive (eavesdropper) antenna, $n_1$, $n_2$ and $m$ are the numbers of Rx, eavesdropper and Tx antennas respectively. The channels ${\rm {\bf H}}_k$ are assumed to be quasistatic (i.e., constant for a sufficiently long period of time so that the infinite horizon information theory assumption holds) and frequency-flat, with full channel state information (CSI) at the Rx and Tx ends.

For a given transmit covariance matrix ${\rm {\bf R}}=E\left\{ {{\rm {\bf xx}}^{\dag}} \right\}$, where $E\left\{ \cdot \right\}$ is the statistical expectation, the maximum achievable secrecy rate between the Tx and Rx (so that the rate between the Tx and the eavesdropper is zero) is \cite{Khisti-2}\cite{Oggier}
\begin{equation}
\label{eq2}
C({\rm {\bf R}})=\ln \frac{\left| {{\rm {\bf I}}+{\rm {\bf W}}_1 {\rm {\bf R}}} \right|}{\left| {{\rm {\bf I}}+{\rm {\bf W}}_2 {\rm {\bf R}}} \right|}=C_1 ({\rm {\bf R}})-C_2 ({\rm {\bf R}})
\end{equation}
where $C_k(\bR)=\ln|\bI +\bW_k\bR|$, $k=1,2$, negative $C({\rm {\bf R}})$ is interpreted as zero rate, ${\rm {\bf W}}_k ={\rm {\bf H}}_k^{\dag} {\rm {\bf H}}_k $, and the secrecy capacity subject to the total Tx power constraint is
\begin{equation}
\label{eq3}
C_s =\mathop {\max }\limits_{{\rm {\bf R}}\ge 0} C({\rm {\bf R}}) \mbox{\ s.t.} \ \tr {\rm {\bf R}}\le P_T
\end{equation}
where $P_T $ is the total transmit power (also the SNR since the noise is normalized). It is well-known that the problem in \eqref{eq3} is not convex in general and explicit solutions for the optimal Tx covariance are not known for the general case, but only for some special cases (e.g. low/high SNR, MISO channels,    full-rank or rank-1 case \cite{Khisti-1}-\cite{Loyka}).

\section{Weak Eavesdropper and Capacity Bounds}
\label{sec:Weak Eavesdropper}

In this section, we obtain novel lower and upper bounds to the secrecy capacity in the general case and show that the bounds are tight when the eavesdropper is weak or if the SNR is low. The weak eavesdropper case is motivated by a scenario where the eavesdropper is located far away from the Tx so that its propagation path loss is large, see e.g. Fig. 2. This is the case when the presence of the eavesdropper does not result in a large capacity loss so that the physical-layer secrecy approach is feasible (while in the case of a strong eavesdropper, the capacity loss is large and other approaches may be preferable, e.g. cryptography).  There is no requirement here for the channel to be degraded or for the optimal covariance to be of full rank or rank 1, so that these results extend the known closed form solutions.

To this end, let
\bal
\label{eq.Cw} \notag
&C_w(\bR) =\ln |\bI + \bW_1 \bR| - \tr (\bW_2 \bR)\\
&C_w = \max_{\bR} C_w(\bR)\\ \notag
&\bR^* = \arg \max_{\bR} C(\bR),\ \bR^*_w = \arg \max_{\bR} C_w(\bR)
\eal
all subject to $\bR \ge 0, \tr\bR \le P_T$, i.e. $\bR^*$ is the optimal covariance and $\bR^*_w$ maximizes $C_w(\bR)$. Using $\ln(1+x)\approx x$ when $0\le x \ll 1$, it can been seen that $C_w(\bR)$ is a weak eavesdropper approximation of $C(\bR)$:
\bal
C(\bR) \approx C_w(\bR)\ \mbox{if}\ \lambda_1(\bW_2\bR) \ll 1
\eal
so that $C_w$ is the weak eavesdropper secrecy capacity. The following Theorem establishes novel secrecy capacity bounds based on $C_w$.
\begin{thm}
The secrecy capacity $C_s$ in \eqref{eq3} for the general Gaussian MIMO-WTC in \eqref{eq1} is bounded as follows:
\bal
\label{eq.T1.1}
C_w \le C(\bR^*_w) \le C_s \le C_w + \frac{P_T^2}{2} \lambda_1^2(\bW_2)
\eal
where
\begin{align}
\label{eq.T1.2}
&\bR^*_w = \bQ^{1/2} (\bI - \widehat{\bW}_1^{-1})_+ \bQ^{1/2}\\
\label{eq.T1.3}
&\widehat{\bW}_1 = \bQ^{1/2} \bW_1 \bQ^{1/2}
\end{align}
and $\bQ$ is the (Moore-Penrose) pseudo-inverse of $\bW_{\lambda}= \lambda\bI + \bW_2$; $\lambda \ge 0$ is found from the total power constraint:
\bal
\label{eq.lambda}
\tr \bR^*_w = P_T \ \mbox{if} \ P_T < P_T^*
\eal
and $\lambda = 0$ otherwise; the threshold power
\bal
\label{eq.T1.5}
P_T^* = \tr \bW_2^{-1}(\bI - \bW_2^{1/2}\bW_1^{-1} \bW_2^{1/2})_+
\eal
if $\bW_2$ is non-singular. When $\bW_2$ is singular, $P_T^* = \infty$ if  $\sN(\bW_2) \nsubseteq \sN(\bW_1)$; otherwise, $\bW_1$ and $\bW_2$ are projected orthogonally to $\sN(\bW_2)$ and the projected matrices are used in \eqref{eq.T1.5}. The weak eavesdropper secrecy capacity can be expressed as
\bal
\label{eq.T1.Cw}
C_w = \sum_{i:\widehat{\lambda}_{1i}>1}\ln\widehat{\lambda}_{1i} - \tr \widehat{\bW}_2(\bI-\widehat{\bW}_1^{-1})_+
\eal
where $\widehat{\lambda}_{1i}=\lambda_i(\widehat{\bW}_1)$, $\widehat{\bW}_2 = \bQ^{1/2} \bW_2\bQ^{1/2}$.
\end{thm}
\bp
See the Appendix.
\ep

\begin{remark}
It may appear that \eqref{eq.T1.2} requires $\widehat{\bW}_1$ and thus $\bW_1$ to be positive definite, i.e. singular case is not allowed. This is not so since $(\cdot)_+$ operator eliminates singular eigenmodes of $\widehat{\bW}_1$ so that  $(\bI-\widehat{\bW}_1^{-1})_+$ is well-defined even if $\bW_1$ is singular: one can use $\widehat{\bW}_{1\delta} = \widehat{\bW}_1 +\delta\bI >0$ instead of $\widehat{\bW}_1$, where $\delta>0$, evaluate $(\bI-\widehat{\bW}_{1\delta}^{-1})_+$ and take the limit $\delta \rightarrow 0$ to see that the singular modes of $\widehat{\bW}_1$ are eliminated so that
\bal
(\bI-\widehat{\bW}_{1}^{-1})_+ = \bU_+\bD\bU_+^{\dag}
\eal
where $\bU_+$ is a semi-unitary matrix whose columns are the eigenvectors of $\widehat{\bW}_1$ corresponding to its positive eigenvalues, $\bD$ is a $r\times r$ diagonal matrix whose $i$-th diagonal entry is $(1-\lam_i^{-1}(\widehat{\bW}_{1}))_+$, $i=1...r$, where $r$ is the rank of $\widehat{\bW}_{1}$. The same observation also applies to \eqref{eq.T1.5}.
\end{remark}

\begin{remark}
The 1st inequality in \eqref{eq.T1.1} bounds the sub-optimality gap of using $\bR^*_w$, for which an achievable rate is $C(\bR^*_w)$, instead of the true optimal covariance $\bR^*$:
\bal
|C_s - C(\bR^*_w)|\le \lam_1^2(\bW_2) P_T^2/2
\eal
so that $C(\bR^*_w) \rightarrow C_s$ as $\lam_1(\bW_2) P_T \rightarrow 0$.
\end{remark}

Using Theorem 1, we can now approximate the secrecy capacity via its weak eavesdropper counterpart.

\begin{cor}
\label{cor.1}
The secrecy capacity of the general Gaussian MIMO-WTC can be expressed as follows:
\bal
\label{eq.Cs.Cw.DC}
C_s = C_w +\Delta C
\eal
where $\Delta C$ is the inaccuracy of the weak eavesdropper approximation, which is bounded as
\bal
\label{eq.DC.bounds}
0 \le \Delta C \le \lam_1^2(\bW_2) P_T^2/2
\eal
so that $\Delta C \rightarrow 0$ and $C_s/C_w \rightarrow 1$ as $P_T \rightarrow 0$ or/and $\lam_1(\bW_2) \rightarrow 0$.
\end{cor}
\begin{proof}
\eqref{eq.Cs.Cw.DC} and \eqref{eq.DC.bounds} follow from the bounds in \eqref{eq.T1.1}, which also implies $\Delta C \rightarrow 0$ as $P_T\lam_1(\bW_2) \rightarrow 0$. To show that $C_s/C_w \rightarrow 1$ as $P_T \rightarrow 0$ observe that
\bal\notag
C_s = P_T \lam_1(\bW_1-\bW_2) + o(P_T) = C_w + o(P_T)
\eal
from which the desired result follows (here, we implicitly assume that $\lam_1(\bW_1-\bW_2)>0$; otherwise, $C_s=0$ and there is nothing to prove). When $\lam_1(\bW_2) \rightarrow 0$, note that both $C(\bR)$ and $C_w(\bR)$ converge to $\ln|\bI+\bW_1\bR|$ so that taking $\max_{\bR}$ results in $C_s/C_w \rightarrow 1$ (since the objectives are continuous and the feasible set is compact).
\end{proof}

Using this Corollary, the secrecy capacity can be approximated as
\bal
\label{eq.Cs.approx.Cw}
C_s \approx C_w
\eal
and the approximation is accurate for a weak eavesdropper or/and low SNR: $\lam_1(\bW_2)P_T \ll 1$, when the bounds in \eqref{eq.T1.1} are also tight, see Fig. 1.

\begin{remark}
Since $\lambda_1(\bW_2\bR) \le \lambda_1(\bW_2)\lambda_1(\bR) \le P_T\lambda_1(\bW_2)$, one way to ensure that the eavesdropper is weak, i.e. $\lambda_1(\bW_2 \bR) \ll 1$ so that $\ln|\bI+\bW_2\bR| \approx \tr \bW_2\bR$, is to require $\lambda_1(\bW_2) \ll 1/P_T$
from which it follows that this holds as long as the power (or SNR) is not too large, i.e. $P_T \ll 1/\lambda_1(\bW_2)$; see also Fig. 1. It should be noted, however, that the approximation in \eqref{eq.Cs.approx.Cw} extends well beyond the low-SNR regime provided that the eavesdropper propagation path loss is sufficiently large (i.e. $\lambda_1(\bW_2)$ is small). For the scenario in Fig. 1, it works well up to about $10$ dB and this can extend to larger SNR for smaller path loss factor $\alpha$.
\end{remark}

\begin{figure}[t]
\label{fig_1}
\centerline{\includegraphics[width=3in]{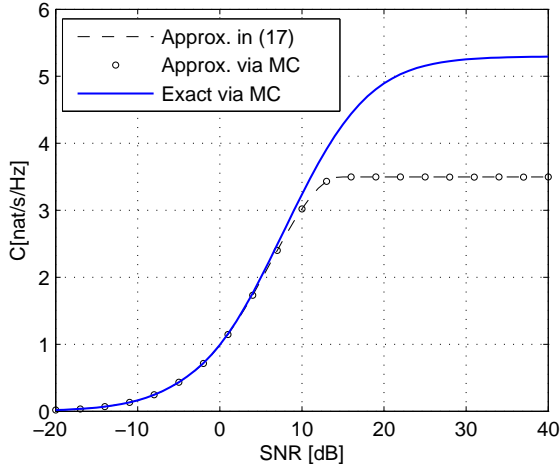}}
\caption{Weak eavesdropper approximation in \eqref{eq.Cs.approx.Cw} and exact secrecy capacity (via MC) versus SNR. $\bW_1$ and $\bW_2$ are as in \eqref{eq.Fig.1}, $\alpha=0.1$, $\lam_1(\bW_2) \approx 0.25$. The approximation is accurate if $\text{SNR} < 10$ dB. Note the capacity saturation effect at high SNR in both cases.}
\end{figure}

To illustrate Theorem 1 and Corrolary 1 and also to see how accurate the approximation is, Fig. 1 shows the secrecy capacity obtained from the approximation in \eqref{eq.Cs.approx.Cw} for
\bal
\label{eq.Fig.1}
\bW_1 = \left(
          \begin{array}{cc}
            2 & 0 \\
            0 & 1 \\
          \end{array}
        \right), \
\bW_2 = \alpha \left(
          \begin{array}{cc}
            2 & 1 \\
            1 & 1 \\
          \end{array}
        \right), \
\eal
also, its exact values (without the weak eavesdropper approximation) obtained by brute force Monte-Carlo (MC) based approach (where a large number of covariance matrices are randomly generated, subject to the total power constraint, and the best one is selected) are shown for comparison. To validate the analytical solution for $C_w$ in Theorem 1, the weak eavesdropper case has also been solved by the MC-based approach. It is clear that the approximation $C_s\approx C_w$ is accurate for the channel in \eqref{eq.Fig.1} provided that $\text{SNR} < 10$ dB. Also note the capacity saturation effect, for both the approximate and exact values. This saturation effect has been already observed in \cite{Khisti-2}\cite{Loyka} and, in the case of $\bW_1>\bW_2>\bf{0}$, the saturation capacity is
\bal
\label{Cs*}
C_s^* = \ln|\bW_1| - \ln|\bW_2|
\eal
which follows directly from \eqref{eq2} by neglecting $\bI$. In the weak eavesdropper approximation, the saturation effect is due to the fact that the 2nd term in \eqref{eq.Cw} is linear in $P_T$ while the 1st one is only logarithmic, so that using the full available power is not optimal when it is sufficiently high. Roughly, the approximation is accurate before it reaches the saturation point, i.e. for $P_T < P_T^*$. The respective saturation capacity is obtained from \eqref{eq.T1.Cw} by setting $\lambda=0$. In the case of $\bW_1>\bW_2>\bf{0}$, it is given by
\bal
\label{C*}
C_w = \ln|\bW_1| - \ln|\bW_2| - \tr (\bI -\bW_2\bW_1^{-1})
\eal
By comparing \eqref{Cs*} and \eqref{C*}, one concludes that the thresholds are close to each other when $\tr\bW_2\bW_1^{-1} \approx m$.

To obtain further insight in the weak eavesdropper regime, let us consider the case when $\bW_1$ and $\bW_2$ have the same eigenvectors. This is a broader case than it may first appear as it requires $\bH_1$ and $\bH_2$ to have the same right singular vectors while leaving left ones unconstrained (see Section \ref{sec:Identical Right Singular Vectors} for more details on this scenario). In this case, the results in Theorem 1 and Corollary 1 simplify as follows.

\begin{cor}
Under the weak eavesdropper condition $\lam_1(\bW_2)\ll 1/P_T$ and when $\bW_1$ and $\bW_2$ have the same eigenvectors, the optimal covariance is
\begin{align}
\label{eq.P1.1}
\bR^* \approx \bR^*_w = \bU\bLam^*\bU^{\dag}
\end{align}
where $\bU$ is found from the eigenvalue decompositions $\bW_i = \bU \bLam_i\bU^{\dag}$ so that the eigenvectors of $\bR^*_w$ are the same as those of $\bW_1$ and $\bW_2$. The diagonal matrix $\bLam^*$ collects the eigenvalues of $\bR^*_w$:
\bal
\label{eq.P1.2}
\lambda_i(\bR^*_w) = \Lb \frac{1}{\lambda+\lambda_{2i}} -\frac{1}{\lambda_{1i}} \Rb_+
\eal
where $\lambda_{ki}$ is $i$-th eigenvalue of $\bW_k$.
\end{cor}
\begin{proof}
Using $\bW_i=\bU\bLam_i\bU^{\dag}$ in \eqref{eq.T1.2} results in \eqref{eq.P1.1} and \eqref{eq.P1.2}.
\end{proof}

Note that the power allocation in \eqref{eq.P1.2} resembles that of the standard water filling, except for the $\lambda_{2i}$ term. In particular, only sufficiently strong eigenmodes are active:
\bal
\label{eq.P1.3}
\lambda_i(\bR^*_w) > 0 \ \mbox{iff} \ \lambda_{1i} > \lambda+\lambda_{2i}
\eal
As $P_T$ increases, $\lambda$ decreases so that more eigenmodes become active; the legitimate channel eigenmodes are active provided that they are stronger that those of the eavesdropper: $\lambda_{1i} > \lambda_{2i}$. Only the strongest eigenmode (for which the difference $\lambda_{1i}-\lambda_{2i}$ is largest) is active at low SNR.

\vspace*{0.5\baselineskip}
\section{Isotropic Eavesdropper and Capacity Bounds}
\label{sec:isotropic}

The model in Section \ref{sec:Weak Eavesdropper} requires the full eavesdropper CSI at the transmitter. This becomes questionable if the eavesdropper does not cooperate (e.g. when it is hidden in order not to compromise its eavesdropping ability). One approach to address this issue is via a compound channel model \cite{Khisti'11}-\cite{Schaefer'13}. An alternative approach is considered here, where the eavesdropper is characterized by its channel gain identical in all directions, which we term "isotropic eavesdropper". This minimizes the amount of CSI available at the transmitter (one scalar parameter and no directional properties).

A further physical justification for this model comes from an assumption that the eavesdropper cannot approach the transmitter too closely due to e.g. some minimum protection distance, see Fig. 2. This ensures that the gain of the eavesdropper channel does not exceed a certain threshold in any transmit direction due to the minimum propagation path loss (induced by the minimum distance constraint). Since the channel power gain in transmit direction $\bf u$ is ${\bf u}^{\dag} {\bf W}_2 {\bf u} = |{\bf H}_2 {\bf u}|^2$ (assuming $|{\bf u}|=1$) and since $\max_{|{\bf u}|=1} {\bf u}^{\dag} {\bf W}_2 {\bf u} = \epsilon_1$ (from the variational characterization of eigenvalues \cite{Horn-1}), where $\epsilon_1$ is the largest eigenvalue of $\bW_2$, ${\bf W}_2 \leq \epsilon_1 {\bf I}$ ensures that the eavesdropper channel power gain does not exceed $\epsilon_1$ in any direction.

In combination with matrix monotonicity of the log-det function,
the latter inequality ensures that $\epsilon_1 {\bf I}$ is the worst possible ${\bf W}_2$ that results in the smallest capacity (the lower bound in \eqref{eq4-2}), i.e. the isotropic eavesdropper with the maximum channel gain is the worst possible one among all eavesdroppers with a bounded spectral norm. Referring to Fig. 2, the eavesdropper channel matrix $\bH_2$ can be presented in the following form:
\bal
\bH_2 = \sqrt{\alpha R_2^{-\nu}}\widetilde{\bH}_2
\eal
where $\alpha R_2^{-\nu}$ represents the average propagation path loss, $R_2$ is the eavesdropper-transmitter distance, $\nu$ is the path loss exponent (which depends on the propagation environment), $\alpha$ is a constant independent of distance (but dependent on frequency, antenna height, etc.) \cite{Rappaport} , and $\widetilde{\bH}_2$ is a properly normalized channel matrix (includes local scattering/multipath effects but excludes the average path loss) so that $\tr\widetilde{\bH}_2^{\dag}\widetilde{\bH}_2 \le n_2 m$ \cite{Loyka-09}. With this in mind, one obtains:
\bal
\label{eq.iso.2}
\bW_2 = \bH_2^{\dag}\bH_2 = \frac{\alpha}{R_2^{\nu}}\widetilde{\bH}_2^{\dag} \widetilde{\bH}_2 \le \frac{\alpha}{R_{2\min}^{\nu}}\widetilde{\bH}_2^{\dag} \widetilde{\bH}_2 \le \frac{\alpha n_2 m}{R_{2\min}^{\nu}}\bI
\eal
so that one can take $\epsilon_1 = \alpha n_2 m R_{2\min}^{-\nu}$ in this scenario, where $R_{2\min}$ is the minimum transmitter-eavesdropper distance. Note that the model captures the impact of the number of transmit and eavesdropper antennas, in addition to the minimum distance and propagation environment. In our view, the isotropic eavesdropper model is more practical than the full Tx CSI model.

\begin{figure}[htbp]
\label{fig.iso}
\centerline{\includegraphics[width=2.5in]{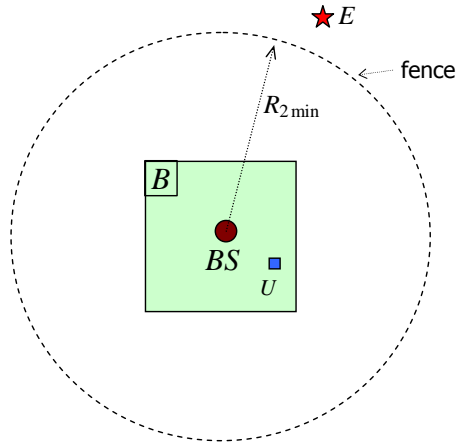}} \caption{Physical scenario for a secret communication system: base station BS (the transmitter) is located on the rooftop of a secure building $B$, legitimate user $U$ (the receiver) is inside the building $B$, and eavesdropper $E$ is beyond the fence so that $R_2 \ge R_{2\min}$.}
\label{fig:Fig_2}
\end{figure}

The isotropic eavesdropper model is closely related to the parallel channel setting in \cite{Khisti-08}\cite{Li-10}: even though the original channel is not parallel, it can be transformed into a parallel channel\footnote{via an information-preserving transformation: using a unitary transmit pre-coding with the unitary matrix whose columns are the eigenvectors of $\bW_1$ and unitary post-codings at the receiver and eavesdropper with unitary matrices whose columns are the left singular vectors of $\bH_1$ and $\bH_2$ respectively.}, for which independent signaling is known to be optimal \cite{Khisti-08}\cite{Li-10}. This shows that signaling on the eigenvectors of $\bW_1$ is optimal in this case while an optimal power allocation is different from the standard water filling  \cite{Li-10}. These properties in combination with the bounds in \eqref{eq4-1} are exploited below.

While it is a challenging analytical task to evaluate the secrecy capacity in the general case, one can use the isotropic eavesdropper model above to construct lower and upper capacity bounds for the general case using the standard matrix inequalities,
\begin{equation}
\label{eq4-1}
 \epsilon_m {\bf I} \leq {\bf W}_2 \leq \epsilon_1 {\bf I}
\end{equation}
where $\epsilon_i = \lambda_i({\bf W}_2)$ denotes $i$-th largest eigenvalue of ${\bf W}_2$, and the equalities are achieved when $\epsilon_1 = \epsilon_m$, i.e. by the isotropic eavesdropper. This is formalized below.

\begin{prop}
\label{prop.isotropic}
The secrecy capacity of the general MIMO-WTC  in \eqref{eq3} is bounded as follows:
\begin{equation}
\label{eq4-2}
 C^*(\epsilon_1)  \leq C_s \leq C^*(\epsilon_m)
\end{equation}
where $C^*(\epsilon)$ is the secrecy capacity if the eavesdropper were isotropic, i.e. under ${\bf W}_2 = \epsilon {\bf I}$,
\begin{eqnarray}
\label{eq4-3} 
C^*(\epsilon) = \mathop {\max }\limits_{{\bf R} \ge 0, \ tr{\bf R} \le P_T} \ln \frac{\left| {{\bf I}+{\bf W}_1 {\bf R}} \right|}{\left| {\bf I}+ \epsilon {\bf R} \right|}
= \sum_i \ln \frac{1+g_i \lambda_i^*}{1+ \epsilon \lambda_i^*}
\end{eqnarray}
$g_i=\lambda_i({\bf W}_1)$, and $\lambda_i^* = \lambda_i({\bf R}^*)$ are the eigenvalues of the optimal transmit covariance ${\bf R}^* = {\bf U}_1 {\bf \Lambda}^*{\bf U}_1^{\dag}$ under the isotropic eavesdropper,
\begin{eqnarray}
\label{eq4-4}
\lambda_i^* = \frac{\epsilon + g_i}{2\epsilon g_i} \left(\sqrt{1+ \frac{4\epsilon g_i}{(\epsilon + g_i)^2} \left( \frac{g_i - \epsilon}{\lambda} - 1 \right)_+} -1 \right)
\end{eqnarray}
and $\lambda > 0$ is found from the total power constraint $\sum_i \lambda_i^* = P_T$.

The gap in the  bounds of \eqref{eq4-2} is upper bounded as follows:
\bal
\label{eq4-4a}
\Delta C = C^*(\epsilon_m) - C^*(\epsilon_1) \le m_+ \ln \frac{1+\epsilon_1 P_T/m_+}{1+\epsilon_m P_T/m_+} \le m_+\ln \frac{\epsilon_1}{\epsilon_m}
\eal
where $m_+$ is the number of eigenmodes such that $g_i > \epsilon_m$. Both bounds are tight at high SNR if $g_{m+} > \epsilon_1$.
\end{prop}
\textbf{Proof:} See the Appendix.

Thus, the optimal signaling for the isotropic eavesdropper case is on the eigenvectors of ${\bf W}_1$ (or right singular vectors of ${\bf H}_1$), identically to the regular MIMO channel, with  the optimal power allocation somewhat similar (but not identical) to the conventional water filling. The latter is further elaborated below for the high and low SNR regimes. Unlike the general case (of non-isotropic eavesdropper), the secrecy capacity of the isotropic eavesdropper case does not depend on the eigenvectors of ${\bf W}_1$ (but the optimal signaling does), only on its eigenvalues, so that the optimal signaling problem here separates into 2 independent parts: (i) optimal signaling directions are selected as the eigenvectors of ${\bf W}_1$, and (ii) optimal power allocation is done based on the eigenvalues of ${\bf W}_1$ and the eavesdropper channel gain $\epsilon$. It is the lack of this separation that makes the optimal signaling problem so difficult in the general case.

The bounds in \eqref{eq4-2} coincide when $\epsilon_1 = \epsilon_m$ thus giving the secrecy capacity of the isotropic eavesdropper. Furthermore, as follows from \eqref{eq4-4a}, they are close to each other when the condition number $\epsilon_1/\epsilon_m$ of ${\bf W}_2$ is not too large, thus providing a reasonable estimate of the capacity, see Fig. 3. Referring to Fig. 2, one can also set $\epsilon_1 = \alpha n_2 m R_{2\min}^{-\nu}$ and proceed with a conservative system design to achieve the secrecy rate $C^*(\epsilon_1)$. Note that this design requires only the knowledge of $n_2$ and $R_{2\min}$ at the transmitter, not full CSI ($\bW_2$) and hence is more realistic. This signaling strategy does not incur significant penalty (compared to the full CSI case) provided that the condition number  $\epsilon_1/\epsilon_m$ is not large, as follows from \eqref{eq4-4a}. It can be further shown that $C^*(\epsilon_1)$ is the compound channel capacity for the class of eavesdroppers with bounded spectral norm (maximum channel gain), $\bW_2 \le \epsilon_1\bI$, and that signaling on the worst-case channel ($\bW_2 = \epsilon_1\bI$) achieves the capacity for the whole class of channels with $\bW_2 \le \varepsilon\bI$ \cite{Schaefer'13}.

We note that the power allocation in \eqref{eq4-4} has properties similar to those of the conventional water-filling, which follow from Proposition \ref{prop.isotropic}.
\begin{prop}
\label{Proposition4-2}
Properties of the optimum power allocation in \eqref{eq4-4} for the isotropic eavesdropper:

1. $\lambda_i^*$ is an increasing function of $g_i$ (strictly increasing unless $\lambda_i^*=0$ or $P_T$) , i.e. stronger eigenmodes get more power (as in the standard WF).

2. $\lambda_i^*$ is an increasing function of $P_T$ (strictly increasing unless $\lambda_i^*=0$). $\lambda_i^* =0$ for $i > 1$ and $\lambda_1^* = P_T$ as $P_T \rightarrow 0$ if $g_1 > g_2$, i.e. only the strongest eigenmode is active at low SNR, and $\lambda_i^* >0$ if $g_i > \epsilon$ as $P_T \rightarrow \infty$, i.e. all sufficiently strong eigenmodes are active at high SNR.

3. $\lambda_i^* > 0$ only if $g_i > \epsilon$, i.e. only the eigenmodes stronger than the eavesdropper ones can be active.

4. $\lambda$ is a strictly decreasing function of $P_T$ and  $0<\lambda < g_1 - \epsilon$; $\lambda \rightarrow 0$ as $P_T \rightarrow \infty$ and $\lambda \rightarrow g_1 - \epsilon$ as $P_T \rightarrow 0$.

5. There are $m_+$ active eigenmodes if the following inequalities hold:
\begin{eqnarray}
\label{eq4-5}
P_{m_+} < P_T \le P_{m_+ +1}
\end{eqnarray}
where $P_{m_+}$ is a threshold power (to have at least $m_+$ active eigenmodes):
\begin{align}
\label{eq4-6}
P_{m_+} = \sum_{i=1}^{m_+ -1} \frac{\epsilon + g_i}{2\epsilon g_i} \left(\sqrt{1+ \frac{4\epsilon g_i}{(\epsilon + g_i)^2} \frac{g_i - g_{m_+}}{(g_{m_+} - \epsilon)_+}} -1 \right),\ m_+ = 2...m,
\end{align}
and $P_1 = 0$, so that $m_+$ increases with $P_T$.
\end{prop}

It follows from Proposition \ref{Proposition4-2} that there is only one active eigenmode, i.e. beamforming is optimal, if $g_2 > \epsilon$ and
\begin{align}
\label{eq4-7}
P_T \le \frac{\epsilon + g_1}{2\epsilon g_1} \left(\sqrt{1+ \frac{4\epsilon g_1}{(\epsilon + g_1)^2} \frac{g_1 - g_2}{g_{2} - \epsilon}} -1 \right)
\end{align}
e.g. in the low SNR regime (note however that the single-mode regime extends well beyond low SNR if $\epsilon \rightarrow g_2$ and $g_1 > g_2$), or at any SNR if $g_{1} > \epsilon$ and $g_{2} \le \epsilon$.

While it is difficult to evaluate $\lambda$ analytically from the power constraint, Property 4 ensures that any suitable numerical algorithm (e.g. Newton-Raphson method) will do so efficiently.

As a side benefit of Proposition \ref{Proposition4-2}, one can use \eqref{eq4-5} as a condition for having $m_+$ active eigenmodes under the regular eigenmode transmission (no eavesdropper) with the standard water-filling by taking $\epsilon \rightarrow 0$ in \eqref{eq4-6}:
\begin{align}
\label{eq4-8}
P_{m_+} = \sum_{i=1}^{m_+ -1} \left( \frac{1}{g_{m_+}} - \frac{1}{g_i}\right)
\end{align}
and \eqref{eq4-8} approximates \eqref{eq4-6} when the eavesdropper is weak, $\epsilon \ll g_{m+}$. To the best of our knowledge, expression \eqref{eq4-8} for the threshold powers of the standard water-filling has not appeared in the literature before.

\subsection{High SNR regime}

Let us now consider the isotropic eavesdropper model when the SNR grows large, so that $g_i \lambda_i^* \gg 1, \epsilon \lambda_i^* \gg 1$. In this case, \eqref{eq4-3} simplifies to

\begin{eqnarray}
\label{eq4-9}
C^*_{\infty} = \sum_{i: g_i > \epsilon} \ln \frac{g_i}{\epsilon}
\end{eqnarray}
where the summation is over active eigenmodes only, so that the capacity is independent of the SNR (saturation effect) and the impact of the eavesdropper is the multiplicative SNR loss, which is never negligible. To obtain a threshold value of $P_T$ at which the saturation takes place, observe that $\lambda\rightarrow 0$ as $P_T \rightarrow\infty$ so that \eqref{eq4-4} becomes
\begin{eqnarray}
\label{eq4-10}
\lambda_i^* = P_T \sqrt{\epsilon^{-1} - g_i^{-1}}/\beta(1+o(1))
\end{eqnarray}
for $i: g_i > \epsilon$, where $\beta =  \sum_{i: g_i > \epsilon} \sqrt{\epsilon^{-1} - g_i^{-1}}$ and $\sqrt{\lambda} = \beta P_T^{-1}(1+o(1))$ from the total power constraint. Using \eqref{eq4-10}, the capacity becomes
\begin{eqnarray}
\label{eq4-12}
C^*(\epsilon) = \sum_{i: g_i > \epsilon} \ln \frac{g_i}{\epsilon} - \frac{\beta^2}{P_T} +o\left(\frac{1}{P_T}\right)
\end{eqnarray}
which is a refinement of \eqref{eq4-9}. The saturation takes place when the second term is much smaller than the first one, so that
\begin{eqnarray}
\label{eq4-13}
P_T \gg \beta^2 / \sum_{i: g_i > \epsilon} \ln \frac{g_i}{\epsilon}
\end{eqnarray}
and $C^*(\epsilon) \approx C^*_{\infty}$ under this condition. This effect in illustrated in Fig. 3. Note that, from \eqref{eq4-10}, the optimal power allocation behaves almost like water-filling in this case, due to the $\sqrt{\epsilon^{-1} - g_i^{-1}}$ term.

Using \eqref{eq4-9}, the gap $\Delta C^*_{\infty}$ between the lower and upper bounds in \eqref{eq4-2} becomes
\begin{eqnarray}
\label{eq4-14} \notag
\Delta C^*_{\infty} &=& C^*_{\infty}(\epsilon_m) - C^*_{\infty}(\epsilon_1) \\
&=& m_1 \ln \frac{\epsilon_1}{\epsilon_m} + \sum_{i=m_1+1}^{m_2} \ln \frac{g_i}{\epsilon_m}
\end{eqnarray}
where $m_1$ and $m_2$ are the numbers of active eigenmodes when $\epsilon = \epsilon_1$ and $\epsilon = \epsilon_m$. Note that this gap is SNR-independent and if $m_1=m_2 = m_+$, which is the case if $g_{m+} > \epsilon_1$, then
\begin{eqnarray}
\label{eq4-15}
\Delta C^*_{\infty} = m_+ \ln \frac{\epsilon_1}{\epsilon_m}
\end{eqnarray}
i.e. also independent of the eigenmode gains of the legitimate user and is determined solely by the condition number of the eavesdropper channel and the number of active eigenmodes. Note that, in this case, the upper bounds in \eqref{eq4-4a} are tight.

\begin{figure}[htbp]
\centerline{\includegraphics[width=3in]{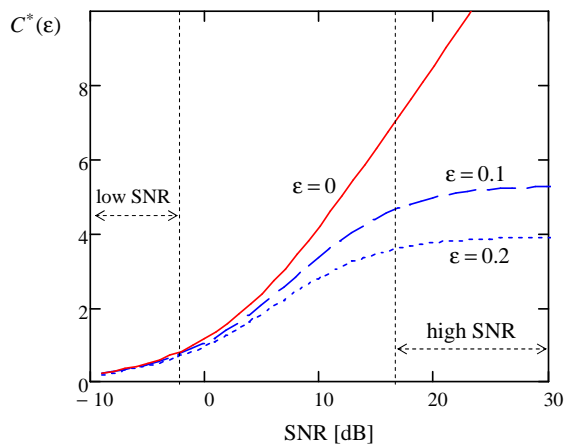}} \caption{Secrecy capacity for the isotropic eavesdropper and the capacity of the regular MIMO channel (no eavesdropper, $\epsilon=0$) vs. the SNR ($=P_T$ since the noise variance is unity); $g_1=2, \ g_2=1$. Note the saturation effect at high SNR , where the capacity strongly depends on $\epsilon$ but not the SNR, and the negligible impact of the eavesdropper at low SNR.}
\label{fig:Fig_3}
\end{figure}

\subsection{When is the eavesdropper's impact negligible?}

It is clear from \eqref{eq4-3} that under fixed $\{g_i\}$ and $P_T$, the secrecy capacity converges to the conventional one $C^*(0)$ as $\epsilon \rightarrow 0$. However, no fixed $\epsilon$ (does not matter how small) can ensure by itself that the eavesdropper's impact on the capacity is negligible since one can always select sufficiently high $P_T$ to make the saturation effect important (see Fig. 3). To answer the question in the section's title, we use \eqref{eq4-3} to obtain:
\begin{align}
\notag
C^*(\epsilon) &= \mathop {\max }\limits_{\{\lambda_i\}} \sum_{i} \ln\left(1+ \frac{1+(g_i-\epsilon)\lambda_i}{1+ \epsilon \lambda_i}\right) \ \mbox{s.t.} \ \lambda_i \ge 0, \sum_i \lambda_i = P_T \\
\label{eq4-16a}
&{\mathop \approx \limits^{(a)}} \mathop {\max} \limits_{\{\lambda_i\}} \sum_{i} \ln (1+(g_i-\epsilon) \lambda_i) \\ \notag
&{\mathop \approx \limits^{(b)}} \mathop {\max }\limits_{\{\lambda_i\}} \sum_{i} \ln (1+g_i \lambda_i) = C^*(0) \end{align}
where (a) holds if
\begin{align}
\label{eq4-17}
P_T \ll 1/\epsilon
\end{align}
(since $\lambda_i \le P_T$), i.e. if the SNR is not too large, and (b) holds if
\begin{align}
\label{eq4-18}
\epsilon \ll g_{i}
\end{align}
for all active eigenmodes, i.e. if the eavesdropper is much weaker than the legitimate active eigenmodes. It is the combination of \eqref{eq4-17} and \eqref{eq4-18} that ensures that the eavesdropper's impact is negligible. Neither condition alone is able to do so. Fig. 3 illustrates this point. Eq. \eqref{eq4-16a} also indicates that the impact of the eavesdropper is the per-eigenmode gain loss of $\epsilon$. Unlike the high-SNR regime in \eqref{eq4-9} where the loss is multiplicative (i.e. very significant and never negligible), here it is additive (mild or negligible in many cases).

\subsection{Low SNR regime}

Let us now consider the low-SNR regime, which is characteristic for CDMA-type systems \cite{Tse}. Traditionally, this regime is defined via $P_T \rightarrow 0$. We, however, use a more relaxed definition requiring that $m_+ =1$, which holds under \eqref{eq4-7}. In this regime, assuming $g_1 > \epsilon$,
\begin{eqnarray}
\label{eq4-19} \notag
C^*(\epsilon) &=& \ln \frac{1+g_1 P_T}{1+\epsilon P_T} = \ln \left(1+ \frac{(g_1- \epsilon )P_T}{1+\epsilon P_T}\right) \\
&{\mathop \approx \limits^{(a)}}& \ln (1+(g_1- \epsilon )P_T)
\end{eqnarray}
where (a) holds when $P_T \ll 1/\epsilon$. It is clear from the last expression that the impact of the eavesdropper is an additive SNR loss of $\epsilon P_T$, which is negligible when $\epsilon \ll g_1$. Note a significant difference to the high SNR regime in \eqref{eq4-9}, where this impact is never negligible. Fig. 3 illustrates this difference.

It follows from \eqref{eq4-19}(a) that the difference between the lower and upper bounds in \eqref{eq4-2} at low SNR is the SNR gap of $(\epsilon_1 -\epsilon_m )P_T$. This difference is negligible if $g_1 \gg \epsilon_1 -\epsilon_m$, which may be the case even if the condition number $\epsilon_1/\epsilon_m$ is large (in which case the difference is significant at high SNR, see \eqref{eq4-15}). Therefore, we conclude that the impact of the eavesdropper is more pronounced in the high-SNR regime and is negligible in the low-SNR one if its channel is weaker than the strongest eigenmode  of the legitimate user, $g_1 \gg \varepsilon_1$.

When $g_1- \epsilon \ll 1/P_T$, \eqref{eq4-19}(a) gives $C^*(\epsilon) \approx (g_1- \epsilon) P_T$, i.e. linear in $P_T$. A similar capacity scaling at low SNR has been obtained in \cite{Rezki} for i.i.d. block-fading single-input single-output (SISO) WTC, without however explicitly identifying the capacity but via establishing upper/lower bounds. Also note that the 1st two equalities in \eqref{eq4-19} do not require $P_T \rightarrow 0$ but only to satisfy \eqref{eq4-7}.

\vspace*{0.5\baselineskip}
\section{Omnidirectional Eavesdropper}
\label{sec:omnidirectional}

In this section, we consider a scenario where the eavesdropper has equal gain in all directions of a certain subspace. This model accounts for 2 points: (i) when the transmitter has no particular knowledge about the directional properties of the eavesdropper, which is most likely from the practical perspective, it is reasonable to assume that its gain is the same in all directions; (ii) on the other hand, when the eavesdropper has a small number of antennas (less than the number of transmit antennas), its channel rank, which does not exceed the number of transmit or receive antennas, is limited by this number so that the isotropic model of the previous section does not apply\footnote{This was pointed out by A. Khisti.}.

For an omnidirectional eavesdropper, its channel gain is the same in all directions of its active subspace, i.e.
\bal
\label{eq.omni.Ev.1}
|\bH_2 \bx|^2  = \bx^{\dag}\bW_2\bx = \mathrm{const} \ \forall \bx \in \mathcal{N}(\bW_2)^{\bot}
\eal
where $\mathcal{N}(\bW_2)^{\bot}$ is the subspace orthogonal to the nullspace $\mathcal{N}(\bW_2)$ of $\bW_2$, i.e. its active subspace, whose dimensionality is $r_2 = \mathrm{rank}(\bW_2)$. In particular, when the eavesdropper is isotropic, $\mathcal{N}(\bW_2)$ is empty so that $\mathcal{N}(\bW_2)^{\bot}$ is the entire space and $r_2=m$. The condition in \eqref{eq.omni.Ev.1} implies that
\bal
\label{eq.omni.Ev.2}
\bW_2 = \varepsilon \bU_{2+}\bU_{2+}^{\dag}
\eal
where $\bU_{2+}$ is a semi-unitary matrix whose columns are the active eigenvectors of $\bW_2$, and $\mathcal{N}(\bW_2)^{\bot} = \mathrm{span}\{\bU_{2+}\}$. Note that the model in \eqref{eq.omni.Ev.2} allows $\bW_2$ to be rank-deficient: $r_2 <m$ is allowed. $\varepsilon$ can be evaluated from e.g. \eqref{eq.iso.2}:  $\varepsilon = \alpha n_2 m R_{2\min}^{-\nu}$.

\begin{thm}
\label{thm.omnidirectional}
Under the omnidirectional eavesdropper setting in \eqref{eq.omni.Ev.1}, \eqref{eq.omni.Ev.2} and when $\mathcal{R}(\bW_1) \subseteq \mathcal{R}(\bW_2)$, the MIMO-WTC secrecy capacity can be expressed as follows:
\bal
\label{eq.omni.Ev.T1} 
 C_{s} = \max_{\tr\bR \le P_T}\ln \frac{| \bI+ \bW_1 \bR |}{|\bf I+ \bW_2 \bf R|} = \max_{\tr \bR \le P_T} \ln \frac{|\bI+\bW_1\bR|}{|\bI+ \epsilon \bR|} = C^*(\epsilon)
\eal
i.e. the capacity and optimal signaling to achieve it are the same as for the isotropic eavesdropper as in Proposition \ref{prop.isotropic}.
\end{thm}
\begin{proof}
First note that, for the omnidirectional eavesdropper, $\bW_2 \le \varepsilon\bI$ so that $|\bI+\bW_2\bR| \le |\bI+\varepsilon\bR|$ and hence
\bal
 C_{s} &= \max_{\tr\bR \le P_T}\ln \frac{| \bI+ \bW_1 \bR |}{|\bf I+ \bW_2 \bf R|} \ge \max_{\tr \bR \le P_T} \ln \frac{|\bI+\bW_1\bR|}{|\bI+ \epsilon \bR|} = C^*(\epsilon)
\eal

To prove the reverse inequality, let $\bP_2$ be a projection matrix on $\mathcal{R}(\bW_2)$, i.e. $\bP_2 = \bU_{2+}\bU_{2+}^{\dag}$. Then, $\bP_2 \bW_k \bP_2 = \bW_k, k=1,2$, so that
\bal
 C(\bR) = \ln \frac{| \bI+ \bP_2 \bW_1 \bP_2 \bR |}{|\bI+ \bP_2 \bW_2 \bP_2 \bf R|} = \ln \frac{| \bI+ \tilde{\bW_1}\tilde{ \bR}|}{|\bI+ \varepsilon\tilde{\bR}|} = \tilde{C}(\tilde{\bR})
\eal
where $\tilde{\bR} = \bU_{2+}^{\dag} \bR \bU_{2+}$ and likewise for $\tilde{\bW_k}$, so that $\tilde{\bW_2} = \varepsilon\bI$, where we used $| \bI+ \bA\bB| = | \bI+ \bB\bA|$. Further note that
\bal
\tr \tilde{\bR} &= \tr \bU_{2+}^{\dag}\bR\bU_{2+} \\
&= \sum_i \lambda_i(\bR) |\bu_{2i}^{\dag} \bu_{Ri}|^2\\
&\le \sum_i \lambda_i(\bR) = \tr \bR \le P_T
\eal
where $\bu_{2i}$ and $\bu_{Ri}$ are $i$-th eigenvectors of $\bW_2$ and $\bR$, and we have used $\bR = \sum_i \lambda_i(\bR) \bu_{Ri}\bu_{Ri}^{\dag}$ and $|\bu_{2i}^{\dag} \bu_{Ri}|^2 \le |\bu_{2i}|^2 |\bu_{Ri}|^2 =1$. Hence, $\tilde{\bR}$ satisfies power constraint if $\bR$ does and thus
\bal
 C_s &= \max_{\tr \bR \le P_T} C(\bR) \le \max_{\tr \tilde{\bR} \le P_T} \tilde{C}(\tilde{\bR}) = \max_{\lambda_i \ge 0,\ \sum_i \lambda_i \le P_T} \sum_i \ln \frac{1+\tilde{g}_i \lambda_i}{1+ \epsilon \lambda_i} = \tilde{C}^*(\varepsilon)
\eal
where $\tilde{g}_i=\lambda_i(\tilde{\bW}_1)$, and $\tilde{C}^*(\varepsilon)$ is the secrecy capacity under $\tilde{\bW}_1$ and isotropic eavesdropper $\tilde{\bW}_2=\varepsilon\bI$. Note that
\bal
\lambda_i(\tilde{\bW_1}) &= \lambda_i(\bU_{2+}^{\dag} \bW_1 \bU_{2+})= \lambda_i([\bU_{2}^{\dag} \bW_1 \bU_{2}]_{r_2 \times r_2}) \le \lambda_i(\bU_{2}^{\dag} \bW_1 \bU_{2}) = \lambda_i(\bW_1)
\eal
where $[\bA]_{k\times k}$ denotes $k\times k$ principal sub-matrix of $\bA$, $r_2 = \mathrm{rank}(\bW_2)$, and $\bU_{2}$ is a unitary matrix whose columns are the eigenvectors of $\bW_2$. The inequality is due to Cauchy eigenvalue interlacing theorem \cite{Horn-1} and the last equality is due to the fact that $\bU_{2} \bW_1 \bU_{2}^{\dag}$ and $\bW_1$ have the same eigenvalues. Based on this, one obtains:
\bal
 C_s &\le \tilde{C}^*(\varepsilon) \le  \max_{\lambda_i \ge 0,\ \sum_i \lambda_i \le P_T} \sum_i \ln \frac{1+g_i \lambda_i}{1+ \epsilon \lambda_i} = C^*(\varepsilon)
\eal
thus establishing $C_{s} = C^*(\varepsilon)$ under an omnidirectional eavesdropper with $\mathcal{R}(\bW_1) \subseteq \mathcal{R}(\bW_2)$.
\end{proof}

Note that the secrecy capacity as well as the optimal signaling for the omnidirectional eavesdropper in Theorem \ref{thm.omnidirectional} is the same as those for the isotropic one (which is not the case in general, as can be shown via examples), i.e. the fact that the rank of the eavesdropper channel is low has no impact provided that $\mathcal{R}(\bW_1) \subseteq \mathcal{R}(\bW_2)$ holds.

Since $\mathcal{R}(\bW)$ collects directions where the channel gain is not zero:
\bal
|\bH \bx|^2  = \bx^{\dag}\bW\bx \neq 0 \ \forall \bx \in \mathcal{R}(\bW)
\eal
the condition $\mathcal{R}(\bW_1) \subseteq \mathcal{R}(\bW_2)$ means that $|\bH_2 \bx|=0$ implies $|\bH_1 \bx|=0$ (but the converse is not true in general) and hence $|\bH_1 \bx|\neq 0$ implies $|\bH_2 \bx|\neq 0$, i.e. the eavesdropper can "see" in any direction where the receiver can "see" (but there is no requirement here for the eavesdropper to be degraded with respect to the receiver so that the channel is not necessarily degraded).

Further note that the condition in \eqref{eq.omni.Ev.1} does not require $\bU_2 = \bU_1$, i.e. the eigenvectors of the legitimate channel and of the eavesdropper can be different.

\vspace*{0.5\baselineskip}
\section{Identical Right Singular Vectors}
\label{sec:Identical Right Singular Vectors}

In this section, we consider the case when $\bH_{1,2}$ have the same right singular vectors (SV), so that their singular value decomposition takes the following form:
\bal
\label{eq.SV.1}
\bH_k = \bU_k \bSigma_k \bV^{\dag}
\eal
where the unitary matrices $\bU_k, \bV$ collect left and right singular vectors respectively and diagonal matrix $\bSigma_k$ collects singular values of $\bH_k$. In this model, the left singular vectors can be arbitrary. This is motivated by the fact that right singular vectors are determined by scattering  around the Tx while left ones - by scattering around the Rx and eavesdropper respectively. Therefore, when the Rx and eavesdropper are spatially separated, their scattering environments may differ significantly (and hence different left SVs) while the same scattering environment around the Tx induces the same right SVs. We make no weak eavesdropper or other assumptions here. After unitary (and thus information-preserving) transformations, this scenario can be put into the parallel channel setting of \cite{Khisti-08}\cite{Li-10}. The secrecy capacity and the optimal covariance in this case can be explicitly characterized as follows.

\begin{prop}
\label{prop.SV}
Consider the wiretap MIMO channel as in \eqref{eq1}, \eqref{eq.SV.1}. The optimal Tx covariance for this channel takes the following form:
\bal
\label{eq.SV.R*}
\bR^* = \bV \bLam^* \bV^{\dag}
\eal
where the diagonal matrix $\bLam^*$ collects its eigenvalues $\lambda^*_i$:
\bal
\label{eq.SV.lambda*}
\lambda^*_i = \frac{\lambda_{2i}+\lambda_{1i}}{2\lambda_{2i}\lambda_{1i}}\Lb \sqrt{1+\frac{4\lambda_{2i}\lambda_{1i}} {(\lambda_{2i}+\lambda_{1i})^2} \Lb \frac{\lambda_{1i}-\lambda_{2i}}{\lambda} - 1 \Rb_+} - 1 \Rb
\eal
and where $\lambda_{ki}=\sigma_{ki}^2$ and $\sigma_{ki}$ denotes singular values of $\bH_{k}$; $\lambda > 0$ is found from the total power constraint: $\sum_i \lambda_i^* = P_T$.
\end{prop}
\begin{proof}
Under \eqref{eq.SV.1}, $\bW_k = \bV \bLam_k \bV^{\dag}$, where diagonal matrix $\bLam_k = \bSigma_k^{\dag}\bSigma_k$ collects eigenvalues of $\bW_k$, so that the problem in \eqref{eq3} can be re-formulated as
\bal
C_s = \max_{\tr \tilde{\bR} \ge \bf{0}} \ln\frac{|\bI+\bLam_1\tilde{\bR}|}{|\bI+\bLam_2\tilde{\bR}|} \ \mbox{s.t.} \ \tr \tilde{\bR} \le P_T
\eal
where $\tilde{\bR} = \bV^{\dag}\bR\bV$. However, this is the secrecy capacity of a set of parallel Gaussian wire-tap channels as in \cite{Khisti-08}\cite{Li-10}, for which independent signaling is known to be optimal\footnote{The authors would like to thank A. Khisti for pointing out this line of argument.}, so that maximizing $\tilde{\bR}^*$ is diagonal, from which \eqref{eq.SV.R*} follows. The optimal power allocation in \eqref{eq.SV.lambda*} is essentially the same as for the equivalent parallel channels in \cite{Li-10}.
\end{proof}

In fact, Eq. \eqref{eq.SV.R*} says that optimal signaling is on the right SVs of $\bH_{1,2}$ and \eqref{eq.SV.lambda*} implies that only those eigenmodes are active for which
\bal
\label{eq.SV.sigma}
\sigma_{1i}^2 > \sigma_{2i}^2 + \lambda
\eal
If $\lambda_{2i}=0$, then \eqref{eq.SV.lambda*} reduces to
\bal
\lambda_i^* = \Lb \frac{1}{\lambda} - \frac{1}{\lambda_{1i}}\Rb_+
\eal
i.e. as in the standard WF. This implies that when $\lambda_{2i}=0$ for all active eigenmodes, then the standard WF power allocation is optimal.

It should be stressed that the original channels in \eqref{eq.SV.1} are not parallel (diagonal). They become equivalent to a set of parallel independent channels after performing information-preserving transformations. Also, there is no assumption of degradedness here and no requirement for the optimal covariance to be of full rank or rank-1.

\vspace*{0.5\baselineskip}
\section{When Is ZF Signaling Optimal?}
\label{sec:ZF}

In this section, we consider the case when ZF signaling is optimal, i.e. when active eigenmodes of the optimal covariance $\bR^*$ are orthogonal to those of $\bW_2$: $\bW_2\bR^*=\bf{0}$\footnote{This simply means that the Tx antenna array puts null in the direction of eavesdropper, which is known as null forming in antenna array literature \cite{VanTrees}. This can also be considered as a special case of interference alignment, so that Proposition \ref{prop.ZF} establishes its optimality.}. It is clear that this does not hold in general. However, the importance of this scenario is coming from the fact that such signaling does not require wiretap codes: since the eavesdropper gets no signal, regular coding on the required channel suffices. Hence, the system design follows the well-established standard framework and secrecy requirement imposes no extra complexity penalty but is rather ensured by the well-established ZF signaling.

\begin{prop}
\label{prop.ZF}
A sufficient condition for Gaussian ZF signaling being optimal for the Gaussian MIMO-WTC in \eqref{eq1} is that $\bW_1$ and $\bW_2$ have the same eigenvectors or, equivalently, $\bH_1$ and $\bH_2$ have the same right singular vectors as in \eqref{eq.SV.1}, and
\bal
\label{eq.ZF.1}
\lambda_{1i} \le \lambda_{2i}+\lambda \ \mbox{if} \ \lambda_{2i} >0,
\eal
where $\lambda$ is found from the total power constraint $\sum_{i} \lambda_i^*=P_T$, and
\bal
\label{eq.ZF.2}
\lambda_i^* = \lambda_i(\bR^*) = \Lb \frac{1}{\lambda} - \frac{1}{\lambda_{1i}}\Rb_+ \ \mbox{if} \ \lambda_{2i}=0,
\eal
and 0 otherwise. The optimal covariance is as in \eqref{eq.SV.R*} so that its eigenvectors are those of $\bW_1$ and $\bW_2$.

A necessary condition of ZF optimality is that the active eigenvectors of $\bR^*$ are also the active eigenvectors of $\bW_1$ and the inactive eigenvectors of $\bW_2$, and that the power allocation is given by \eqref{eq.ZF.2}.

\end{prop}
\begin{proof}
See the Appendix.
\end{proof}

\begin{remark}
The optimal power allocation in \eqref{eq.ZF.2} is the same as standard water filling. However, a subtle difference here is the condition for an eigenmode to be active, $\lambda_i^* >0$: while the standard WF requires $\lambda_{1i} > \lambda$, the solution above requires in addition $\lambda_{2i}=0$, so that the set of active eigenmodes is generally smaller: the larger the set of eavesdropper positive eigenmodes, the smaller the set of active eigenmodes.
\end{remark}

It is gratifying to see that the standard WF over the eigenmodes of the required channel is optimal if ZF is optimal. In a sense, the optimal transmission strategy in this case is separated into two independent parts: part 1 ensures that the eavesdropper gets no signal (via the ZF) and part 2 is the standard eigenmode signaling and WF on what remains of the required channel as if the eavesdropper were not there. No new wiretap codes need to be designed.

\section{When Is the Standard Water Filling Optimal?}
\label{sec:WF}

Motivated by the fact that the transmitter may be unaware about the presence of an eavesdropper and hence uses the standard transmission on the eigenmodes of $\bW_1$ with power allocated via the water-filling (WF) algorithm, we ask the question: is it possible for this strategy to be optimal for the MIMO-WTC? The affirmative answer and conditions for this to happen are given below.
To this end, let $\bR_{WF}$ be the optimal Tx covariance matrix for transmission on $\bW_1$ only, which is given by the standard water-filling over the eigenmodes of $\bW_1$:
\begin{align}
\label{eq.WF}
\bR_{WF} &= \bU_1 \bLam^*\bU_1^{\dag}, \ \lambda_i^*= \left\{\lambda^{-1} - \lambda_{1i}^{-1}\right\}_+
\end{align}
where $\bLam^*=diag\{\lambda_i^*\}$ is a diagonal matrix of the eigenvalues of $\bR_{WF}$, and $\lambda$ is found from the total power constrain $\sum_i \lambda_i^*=P_T$.

\begin{thm}
\label{thm.WF}
The standard WF Tx covariance matrix in \eqref{eq.WF} is also optimal for the Gaussian MIMO-WTC if:

1) the eigenvectors of $\bW_1$ and $\bW_2$ are the same: $\bU_{1}=\bU_{2}$;

2) for active eigenmodes $\lambda_i^*>0$, their eigenvalues $\lambda_{1i}$ and $\lambda_{2i}$ are related as follows:
\bal
\label{eq.thm.WF.1}
\lambda_{2i} = \frac{\lambda_{1i}}{1+\alpha\lambda_{1i}} < \lambda_{1i},\ \mbox{for some}\ \alpha>0,
\eal
or, equivalently, $\lambda_{2i}^{-1}=\lambda_{1i}^{-1}+\alpha$;

3) for inactive eigenmodes $\lambda_i^*=0$, the eigenvalues $\lambda_{1i}$ and $\lambda_{2i}$ are related either as in \eqref{eq.thm.WF.1} or $\lambda_{1i} \le \lambda_{2i}$.

\end{thm}
\begin{proof}
We assume that $\bW_1$ and $\bW_2$ are non-singular; the singular case will be considered below (using a standard continuity argument). The KKT conditions for the optimal covariance $\bR=\bR_{WF}$, which are necessary for optimality in \eqref{eq3}, can be expressed as:
\bal
\label{eq.WF.KKT.1a}
(\bW_1^{-1}+\bR)^{-1}-(\bW_2^{-1}+\bR)^{-1}=\lambda'\bI-\bM\\
\label{eq.WF.KKT.1b}
\lambda'(\tr\bR-P_T)=0,\ \bM\bR=0 \\
\label{eq.WF.KKT.1c}
\lambda' \ge 0,\ \bM, \bR \ge 0,\ tr\bR \le P_T
\eal
where $\bM \ge 0$ is the Lagrange multiplier matrix responsible for the constraint $\bR \ge 0$ while $\lambda' \ge 0$ is the Lagrange multiplier responsible for the total power constraint $tr\bR \le P_T$. Multiplying both sides of \eqref{eq.WF.KKT.1a} by $\bU_1^{\dag}$ on the left and by $\bU_1$ on the right, one obtains:
\bal
\label{eq.WF.KKT.2}
(\bLam_1^{-1}+\bLam^*)^{-1}-(\bLam_2^{-1}+\bLam^*)^{-1} =\lambda'\bI- \bU_1^{\dag}\bM\bU_1 = \lambda'\bI-\bLam_M
\eal
where $\bLam_1, \bLam_2, \bLam_M$ are diagonal matrices of eigenvalues of $\bW_1, \bW_2, \bM$. The last equality follows from the fact that all terms but $\bU_1^{\dag}\bM\bU_1$ are diagonal so that the last term has to be diagonal too: $\bU_1^{\dag}\bM\bU_1 = \bLam_M$, i.e. $\bM$ has the same eigenvectors as $\bW_1, \bW_2, \bR$. The complementary slackness in \eqref{eq.WF.KKT.1b} implies that $\lambda_i^* \lambda_{Mi} =0$, where $\lambda_{Mi}$ is $i$-th eigenvalue of $\bM$, i.e. if $\lambda_i^*>0$ (active eigenmode) then $\lambda_{Mi}=0$ so that, after some manipulations, \eqref{eq.WF.KKT.2} can be expressed as
\bal \notag
\label{eq.WF.KKT.2a}
\lambda_i^* = \frac{1}{(\lambda_{2i}^{-1}+\lambda_i^*)^{-1} + \lambda'} - \frac{1}{\lambda_{1i}} =\lambda^{-1} - \lambda_{1i}^{-1}
\eal
for each $\lambda_i^*>0$, where the 2nd equality follows from \eqref{eq.WF}. Therefore, $\lambda =  (\lambda_{2i}^{-1}+\lambda_i^*)^{-1} + \lambda'$ and hence
\bal 
\lambda_i^* =(\lambda-\lambda')^{-1} - \lambda_{2i}^{-1} = \lambda^{-1} - \lambda_{1i}^{-1}
\eal
so that $\lambda_{2i}^{-1} = \lambda_{1i}^{-1} + \alpha$ with $\alpha = (\lambda-\lambda')^{-1} - \lambda^{-1}>0$ satisfies both equalities in \eqref{eq.WF.KKT.2a}.

For inactive eigenmodes $\lambda_i^*=0$, it follows from \eqref{eq.WF.KKT.2} that
\bal 
\label{eq.WF.KKT.5}
\lambda_{1i} - \lambda_{2i} = \lambda' - \lambda_{Mi} \le \lambda'
\eal
Observe that this inequality is satisfied when $\lambda_{1i} \le \lambda_{2i}$ (since $\lambda' >0$). To see that it also holds under \eqref{eq.thm.WF.1}, observe that
\bal 
\lambda_{1i} - \lambda_{2i} = \frac{\alpha\lambda_{1i}^2}{1+\alpha\lambda_{1i}} \le \frac{\alpha\lambda^2}{1+\alpha\lambda} = \lambda'
\eal
where the inequality is due to $\lambda_{1i} \le \lambda$ (which holds for inactive eigenmodes) and the fact that $\frac{\alpha\lambda_{1i}^2}{1+\alpha\lambda_{1i}}$ is increasing in $\lambda_{1i}$. Thus, one can always select $\lambda_{Mi} \ge 0$ to satisfy \eqref{eq.WF.KKT.5} and hence the KKT conditions in \eqref{eq.WF.KKT.1a}-\eqref{eq.WF.KKT.1c} have a unique solution which also satisfies \eqref{eq.WF}. This proves the optimality of $\bR_{WF}$.

If $\bW_1$ or/and $\bW_2$ are singular, one can use a standard continuity argument: observe that $C_s$ is a continuous function of $\bW_1$ and $\bW_2$ (which follows from the continuity of $C(\bR)$ and the compactness of the constraint set $\{\bR: \bR \ge 0, \tr\bR \le P_T\}$, which is closed and bounded) and that the conditions 1-3 of Theorem 3 are also continuous. Hence, one can consider $\bW_{k\delta}= \bW_k +\delta\bI>0$, where $\delta >0$ and $k=1,2$, instead of $\bW_k$, apply Theorem 3 and then take the limit $\delta \rightarrow 0$ to establish the result for the singular case.
\end{proof}

Note that the conditions of Theorem \ref{thm.WF} do not require $\bW_1 = a \bW_2$ for some scalar $a >1$; they also allow for the WTC to be non-degraded. However, the condition in \eqref{eq.thm.WF.1} implies that larger $\lambda_{1i}$ corresponds to larger $\lambda_{2i}$, so that, over the active signaling subspace, the channel is degraded.

The 1st condition in Theorem \ref{thm.WF} implies that $\bH_1$ and $\bH_2$ have the same right singular vectors but imposes no constraints on their left singular vectors. This may represent a scenario where the transmitter is a basestation  where the legitimate channel and the eavesdropper experience the same scattering  while having their own individual scatterers around their own receivers (which determine the left singular vectors), as in Section \ref{sec:Identical Right Singular Vectors}.

\vspace*{0.5\baselineskip}
\section{When Is Isotropic Signaling Optimal?}
\label{sec:isotropic R}

In the regular MIMO channel ($\bf W_2 = 0$), the isotropic signaling (IS) is optimal (${\bf R}^\ast = a {\bf I}$) iff ${\bf W}_1 = b {\bf I}$, i.e. ${\bf W}_1$ has identical eigenvalues. Since this transmission strategy is appealing due to its low complexity (all antennas send independent data streams, no precoding, no Tx CSI and thus no feedback is required), we consider the isotropic signaling over the wire-tap MIMO channel and characterize the set of channels on which it is optimal. It turns out to be much richer than that of the regular MIMO channel.

\begin{prop}
\label{Proposition6-1}
Consider the MIMO wire-tap channel in \eqref{eq1}. The isotropic signaling is optimal, i.e. ${\bf R}^\ast = a {\bf I}$ in \eqref{eq3}, for the set of channels $\{{\bf W}_1, {\bf W}_2\}$ that satisfy all of the following:

1. ${\bf W}_1$ and ${\bf W}_2$ have the same (otherwise arbitrary) eigenvectors, ${\bf U}_1 = {\bf U}_2$.

2. ${\bf W}_1 > {\bf W}_2$ so that $\lambda_i ({\bf W}_1) = a_i^{-1} > \lambda_i({\bf W}_2) = b_i^{-1}$, where $\lambda_i ({\bf W})$ are ordered eigenvalues of ${\bf W}$.

3. Take any $b_1 >0$ and $a_1 < b_1$ and set $\lambda = (a_1 + a)^{-1} - (b_1 + a)^{-1} > 0$,

4. For $i=2...m$, take any $b_i$ such that $b_i > \lambda a^2 (1-\lambda a)^{-1} > 0$,
and set
\begin{align}
a_i = -a + (\lambda + (b_i + a)^{-1})^{-1} > 0
\end{align}

This gives the complete characterization of the set of channels for which isotropic signaling is optimal.
\end{prop}
\begin{proof}
It is straightforward to see that any channel in the given set satisfies the conditions of Theorem 2 in \cite{Loyka} and the corresponding optimal covariance is isotropic, which proves the sufficiency. The converse (necessity)  follows from Theorem 1 in \cite{Loyka}, which requires  ${\bf W}_1 > {\bf W}_2$, so that the optimization problem is strictly convex and thus has a unique solution. For isotropic signaling to be optimal, the corresponding KKT conditions (see the proofs of Theorems 1 and 2 in \cite{Loyka}) imply the conditions stated above.
\end{proof}
\vspace*{0.5\baselineskip}

Note that the special case of this Proposition is when ${\bf W}_1$ and ${\bf W}_2$ have identical eigenvalues, as in the case of the regular MIMO channel, but, unlike the regular channel, there is also a large set of channels with distinct eigenvalues which dictates the isotropic signaling as well. It is the interplay between the legitimate user and the eavesdropper that is responsible for this phenomenon, i.e. a non-isotropic nature of the 1st channel is compensated for by a carefully-adjusted non-isotropy of the 2nd one.

Table 1 summarizes the conditions for the optimality of the ZF, the WF and the IS in the Gaussian MIMO-WTC. Clearly, the requirement for $\bW_1$ and $\bW_2$ to have the same eigenvectors is the key condition. It is satisfied when the legitimate receiver and the eavesdropper are subject to the same scattering around the base station (the transmitter) while they may have their own sets of scatterers around their own units.

\begin{table}
\caption{The conditions of optimality of the ZF, the WF and the IS in the Gaussian MIMO-WTC}
\label{tab.1}
\centering
\begin{tabular}{|c|c|}
\hline
\bfseries Strategy & \bfseries Optimality conditions\\
  \hline \hline
  WF & $\bU_1=\bU_2$; $\lam_{1i}, \lam_{2i}$ as in Theorem 3 \\
  \hline
  ZF & $\bU_1=\bU_2$; $\lam_{1i}, \lam_{2i}$ as in Proposition 4 \\
  \hline
  IS & $\bU_1=\bU_2$; $\lam_{1i}, \lam_{2i}$ as in Proposition 5 \\
  \hline
\end{tabular}
\end{table}

\vspace*{0.5\baselineskip}
\section{Acknowledgement}
The authors would like to thank M. Urlea and K. Li for running numerical experiments and generating Fig. 1, and A. Khisti for suggesting the problem formulation in Section \ref{sec:omnidirectional}.

\appendix

\subsection{Proof of Theorem 1}

Applying the inequalities
\bal
\label{eq.ln.ineq}
x - x^2/2 \le \ln(1+x) \le x
\eal
which hold for any $x\ge 0$, to
\bal
\ln|\bI+\bW_2\bR| = \sum_i \ln(1+\lam_i(\bW_2\bR))
\eal
one obtains:
\bal
\label{eq.ineq.Cw}
C_w(\bR) \le C(\bR) \le C_w(\bR) + \frac{1}{2} \sum_i \lambda_i^2(\bW_2\bR)
\eal
from which the 1st inequality in \eqref{eq.T1.1} follows by using $\bR=\bR^*_w$; the 2nd inequality follows from the fact that $C(\bR)$ is maximized by $\bR^*$: $C_s = C(\bR^*) \ge C(\bR_w^*)$. To obtain the last inequality, we need the following lemma.

\begin{lemma}
Let $\lam_i \ge 0$ and $\sum_i \lam_i \le P_T$. Then,
\bal
\sum_i \lam_i^2 \le P_T^2
\eal
\end{lemma}
\begin{proof}
Since $\lam_i \ge 0$,
\bal
\sum_i \lam_i^2 \le \left(\sum_i \lam_i\right)^2 \le P_T^2
\eal
\end{proof}

Using this Lemma and observing that $\lam_i(\bW_2\bR) \le \lam_1(\bW_2)\lam_i(\bR)$ (see e.g. \cite{Horn-1}), one obtains:
\bal
\sum_i \lam_i^2(\bW_2\bR) \le \lam_1^2(\bW_2)\sum_i\lam_i^2(\bR) \le \lam_1^2(\bW_2)P_T^2
\eal
since $\sum_i \lam_i(\bR) \le P_T$, so that
\bal
C_s = C(\bR^*) \le C_w(\bR^*) + \lam_1^2(\bW_2)P_T^2/2 \le C_w + \lam_1^2(\bW_2)P_T^2/2
\eal
since $C_w =C_w(\bR^*_w) \ge C_w(\bR^*)$, which establishes the last inequality in \eqref{eq.T1.1}.

To establish the closed form solution for $C_w$ in \eqref{eq.T1.Cw}, consider the optimization problem in \eqref{eq.Cw}, for which the Lagrangian is
\bal
L = \ln |\bI + \bW_1 \bR| - \tr (\bW_2 \bR) -\lambda(\tr\bR-P_T) + \tr (\bM\bR)
\eal
where $\lambda \ge 0$ is a Lagrange multiplier responsible for the total power constraint and $\bM \ge \bf{0}$ is a matrix Lagrange multiplier responsible for the constraint $\bR \ge \bf{0}$. The corresponding KKT conditions (see e.g. \cite{Boyd} for a background on these conditions) are:
\bal
\label{eq.KKT.1}
\partial L /\partial\bR=(\bI+\bW_1\bR)^{-1}\bW_1-\bW_2-\lambda\bI+\bM = \bf{0} \\
\lambda(\tr\bR-P_T)=0, \bM\bR=\bf{0} \\
\lambda \ge 0, \bM, \bR \ge \bf{0}
\eal
Since the objective is concave, the corresponding optimization problem is convex, and since Slater condition holds (e.g. take $\bR=P_T\bI/2 >0, tr\bR < P_T$), the KKT conditions are sufficient for optimality \cite{Boyd}. After some manipulations, \eqref{eq.KKT.1} can be transformed to
\bal
\label{eq.KKT.1a}
&\widehat{\bR}-(\bI -\widehat{\bM})^{-1} = -\widehat{\bW}_1^{-1}\\
\label{eq.T1.tilde}
&\widehat{\bR}=\bW_{\lambda}^{1/2}\bR\bW_{\lambda}^{1/2}, \ \widehat{\bM}=\bW_{\lambda}^{-1/2}\bM\bW_{\lambda}^{-1/2},\ \widehat{\bW}_1=\bW_{\lambda}^{-1/2}\bW_1\bW_{\lambda}^{-1/2}
\eal
where we implicity assume that $\bW_1$ and $\bW_{\lam}$ are non-singular, so that $\bQ=\bW_{\lambda}^{-1}$; the singular case will be considered below. Since $\widehat{\bM}\widehat{\bR}=\bf{0}$ (which follows from $\bM\bR=\bf{0}$), these matrices commute and thus have the same eigenvectors, which, from \eqref{eq.KKT.1a}, implies that these eigenvectors are the same as those of $\widehat{\bW}_1$. Hence, all three matrices can be simultaneously diagonalized and thus \eqref{eq.KKT.1a} can be transformed to diagonal form where the diagonal entries are respective eigenvalues:
\bal
\label{eq.KKT.1a2}
\lambda_i(\widehat{\bR})-(1 -\lambda_i(\widehat{\bM}))^{-1} = -\lambda_i^{-1}(\widehat{\bW}_1)
\eal
From this and complementary slackness $\widehat{\bM}\widehat{\bR}=\bf{0}$, which implies $\lambda_i(\widehat{\bM})=0$ if $\lambda_i(\widehat{\bR})>0$ (i.e. for active eigenmodes),
\bal
\label{eq.lam.R.bar}
\lam_i(\widehat{\bR}) = (1-\lam_i^{-1}(\widehat{\bW}_1))_+
\eal
so that $\widehat{\bR} = (\bI-\widehat{\bW}_1^{-1})_+$ from which \eqref{eq.T1.2} follows. Lagrange multiplier $\lambda$ is found from the total power constraint $tr\bR \le P_T$.

The existence of the threshold power $P_T^*$ follows from the fact that $\tr\bR^*$ is monotonically decreasing in $\lambda$ so that its largest value corresponds to $\lambda\rightarrow 0$ and equals $P_T^*$. When $P_T > P_T^*$, $\lambda = 0$ and $\tr \bR^* = P_T^* < P_T$, i.e. only partial power is used (see Fig. 1 for illustration and discussion). The fact that $P_T^* = \infty$ if $\bW_2$ is singular  and $\sN(\bW_2) \nsubseteq \sN(\bW_1)$ can be established via a limiting transition: consider $\bW_{2\delta} = \bW_2 + \delta\bI >0$ instead of $\bW_2$, where $\delta >0$, evaluate $P_T^*(\delta)$ and take the limit $\lim_{\delta\rightarrow 0} P_T^*(\delta)$ ($P_T^* = \infty$ corresponds to the fact that one can always use extra power to transmit on the directions in $\sN(\bW_2)$ for which there is no leakage to the eavesdropper but positive rate to the legitimate receiver). If $\sN(\bW_2) \subseteq \sN(\bW_1)$, one can project both matrices orthogonaly to the subspace $\sN(\bW_2)$ without affecting the system performance, and perform the analysis on the projected matrices (of which the projected $\bW_2$ is non-singular).

If $\bW_{\lam}$ is singular, it follows from \eqref{eq.KKT.1} that $\lam = 0$ (inactive total power constraint) and $\bW_1$ is singular as well and, furthermore, $\sN(\bW_2) \subseteq \sN(\bW_1)$ so that both matrices can be projected, without affecting the performance, on the subspace orthogonal to $\sN(\bW_2)$, the analysis can be carried out for the projected matrices (where the projected $\bW_2$ is non-singular), and the resulting covariance can be transformed back to the original space. This is equivalent to using the (Moore-Penrose) pseudo-inverse $\bQ$ of $\bW_{\lam}$ instead the inverse in \eqref{eq.T1.2} and \eqref{eq.T1.3}. This approach can also be used to compute the threshold power $P_T^*$ if $\bW_2$ is singular and $\sN(\bW_2) \subseteq \sN(\bW_1)$. The case of singular $\bW_1$ is also addressed in Remark 1.

Finally, \eqref{eq.T1.Cw} is obtained by using \eqref{eq.T1.2} in \eqref{eq.Cw}.

\subsection{Proof of Proposition \ref{prop.isotropic}}
\label{sec:Proof Prop4-1}
The 1st equality in \eqref{eq4-3} follows from \eqref{eq3}. The 2nd equality follows from the Hadamard inequality applied to $| {\bf I}+{\bf W}_1 {\bf R}|$ in the same way as for the regular MIMO channel, and the equality is achieved  when ${\bf R}$ has the same eigenvectors as ${\bf W}_1$, ${\bf R}^* = {\bf U}_1 {\bf \Lambda}^*{\bf U}_1^{\dag}$, which maximizes the numerator and leaves the denominator unchanged. The remaining part is the optimal power allocation in \eqref{eq4-4}, which can be formulated as
\begin{eqnarray}
\label{eqF1}
C^*(\epsilon) = \mathop {\max}\limits_{\{\lambda_i\}} \sum_i \ln \frac{1+g_i \lambda_i}{1+ \epsilon \lambda_i}, \ \mbox{s.t.} \ \lambda_i \ge 0, \sum_i \lambda_i = P_T
\end{eqnarray}
This, however, represents an optimal power allocation for parallel channels which can be found in \cite{Li-10}.

The lower/upper bounds follow from the fact that $| {\bf I}+{\bf W} {\bf R}|$ is a matrix-monotone function of $\bf W$ \cite{Horn-1}, so that $| {\bf I}+{\bf W}_b {\bf R}| \ge | {\bf I}+{\bf W}_a {\bf R}|$ $\forall \ {\bf W}_b \ge {\bf W}_a \ge \bf 0$.

To establish the gap bound in \eqref{eq4-4a}, observe the following:
\bal
\label{eqF5}
\Delta C &= C^*(\epsilon_m) - C^*(\epsilon_1) = \max_{\{\lambda_i\}} \sum_{i} \ln\frac{1+g_i\lambda_i}{1+\epsilon_m\lambda_i} - \max_{\{\lambda_i\}} \sum_{i} \ln\frac{1+g_i\lambda_i}{1+\epsilon_1\lambda_i} \\
\label{eqF5b}
&\le \max_{\{\lambda_i\}} \sum_{i: g_i>\epsilon_m} \ln\frac{1+\epsilon_1\lambda_i}{1+\epsilon_m\lambda_i}\\
\label{eqF5c}
&= m_+ \ln \frac{1+\epsilon_1 P_T/m_+}{1+\epsilon_m P_T/m_+} \\
\label{eqF5d}
&\le m_+\ln \frac{\epsilon_1}{\epsilon_m}
\eal
where maximization is over the set of positive $\{\lambda_i\}$ satisfying the power constraint $\sum_i \lambda_i \le P_T$, and $m_+$ is the number of active eigenmodes. \eqref{eqF5b} follows from (easy to verify) fact that
\bal
\max_x f(x) - \max_x g(x) \le \max_x \{f(x)-g(x)\}
\eal
and the observation that the 1st maximization in \eqref{eqF5} requires $g_i>\epsilon_m$ for any $\lam_i >0$ so that imposing the same condition on the 2nd maximization results in an upper bound.
To show \eqref{eqF5c}, observe that the sum in \eqref{eqF5b} is permutation-symmetric, i.e. has the same value for $\blam = [\lambda_1,...,\lambda_{m+}]$ and any of its permutation $\pi_k\{\blam\}$, where $\pi_k$ denotes a permutation. Let $F(\blam)$ be this sum and observe further that it is concave in $\blam$ (since each term is), so that
\bal
F(\blam) = \frac{1}{m_+!}\sum_k F(\pi_k\{\blam\}) \le F\Lb \frac{1}{m_+!}\sum_k \pi_k\{\blam\} \Rb \le F\Lb \left\{\frac{P_T}{m_+}\right\}\Rb
\eal
where $\{P_T/m_+\}$ is a vector with all entries equal to $P_T/m_+$. The 1st equality is due to permutation symmetry, the 1st inequality is due to the concavity of $F(\blam)$, and last inequality is due to the power constraint and the fact that $F(\blam)$ is increasing in each $\lambda_i$. Since this holds for each $\blam$ (including optimal one), \eqref{eqF5c} follows. \eqref{eqF5d} follows from the fact that \eqref{eqF5c} is monotonically increasing in $P_T$.

\subsection{Proof of Proposition 4}

The original problem in \eqref{eq3} is not convex in general. However, since the objective is continuous, the feasible set is compact and Slater condition holds, KKT conditions are necessary for optimality \cite{Berstekas}. They take on the following form (see e.g. \cite{Loyka}):
\bal
\label{eq.ZF.KKT.1a}
\lambda\bW_1\bR = \bW_1-\bW_2+\bM-\lambda\bI\\
\label{eq.ZF.KKT.1b}
\lambda(\tr\bR-P_T)=0,\ \bM\bR=0 \\
\label{eq.ZF.KKT.1c}
\lambda \ge 0,\ \bM, \bR \ge 0,\ \tr\bR \le P_T
\eal
where $\bM \ge 0$ is the Lagrange multiplier matrix responsible for the constraint $\bR \ge 0$ while $\lambda \ge 0$ is the Lagrange multiplier responsible for the total power constraint $tr\bR \le P_T$, and we used the orthogonality condition $\bW_2\bR=0$.

To prove sufficiency, note from Proposition \ref{prop.SV} that if $\bW_1, \bW_2$ have the same eigenvectors so is $\bR$ and hence $\bM$ and also the KKT conditions are sufficient for optimality (since they have a unique solution). Hence, \eqref{eq.ZF.KKT.1a} can be transformed to a diagonal form:
\bal
\lambda\lambda_{1i}\lambda_i = \lambda_{1i}-\lambda_{2i}+ \lambda_{Mi} -\lambda
\eal
where $\lambda_i, \lambda_{Mi}$ are the eigenvalues of $\bR, \bM$. Complementary slackness in \eqref{eq.ZF.KKT.1b} gives $\lambda_i \lambda_{Mi} =0$ so that $\lambda_i>0$ (active eigenmodes) implies $\lambda_{Mi}=0$ and hence
\bal
\lambda_i = \frac{\lambda_{1i} - \lambda_{2i} - \lambda} {\lambda\lambda_{1i}} = \frac{1}{\lambda} - \frac{1}{\lambda_{1i}}
\eal
where the 2nd equality follows from the orthogonality condition $\lambda_{2i} \lambda_{i}=0$. For inactive eigenmodes $\lambda_i=0$, one obtains $\lambda_{Mi} = \lambda - \lambda_{1i} + \lambda_{2i} \ge 0$ so that $\lambda_{1i} \le \lambda + \lambda_{2i}$.

To prove the necessary part, note that complementary slackness $\bR\bM=0$ implies that $\bR\bM=\bM\bR$ and hence $\bR, \bM$ have the same eigenvectors so that the eigenvalue decompositions are: $\bR=\bU\bLam\bU^{\dag}, \bM=\bU\bLam_M\bU^{\dag}$, where diagonal matrices $\bLam, \bLam_M$ collect respective eigenvalues, and the columns of unitary matrix $\bU$ are the eigenvectors. Multiplying \eqref{eq.ZF.KKT.1a} by $\bU^{\dag}$ from the left and by $\bU$ from the right, one obtains, after some manipulations,
\bal
\label{eq.ZF.3}
\lambda\bI - \bLam_M = \widetilde{\bW}_1(\bI - \lambda\bLam) - \widetilde{\bW}_2
\eal
where $\widetilde{\bW}_k = \bU^{\dag}\bW_k\bU$. Using the orthogonality condition $\bR\bW_2 = \bW_2\bR=0$, which imply $\widetilde{\bW}_2\bLam =0$, and block-partitioned representation of $\bLam, \widetilde{\bW}_2$, one obtains:
\bal
\widetilde{\bW}_2\bLam =
\underbrace{\left(
  \begin{array}{cc}
    \bA_{11} & \bA_{12} \\
    \bA_{21} & \bA_{22} \\
  \end{array}
\right) }_{\widetilde{\bW}_2}
\underbrace{\left(
  \begin{array}{cc}
    \bLam_r & 0 \\
     0 & 0 \\
  \end{array}
\right)}_{\bLam} =
\left(
  \begin{array}{cc}
    \bA_{11}\bLam_r & 0 \\
    \bA_{21}\bLam_r & 0 \\
  \end{array}
\right) = 0
\eal
where diagonal matrix $\bLam_r>0$ collects positive eigenvalues of $\bR$, so that $\bA_{11}=0, \bA_{21}=\bA_{12}^{\dag}=0$ and hence $\widetilde{\bW}_2$ is block-diagonal: $\widetilde{\bW}_2 = diag\{0, \bA_{22}\}$. This proves that active eigenvectors of $\bR$ are also inactive eigenvectors of $\bW_2$.
Complementary slackness $\bR\bM=0$ implies $\bLam\bLam_M=0$ so that $\bLam_M$ is also block-diagonal: $\bLam_M= diag\{0, \bLam_{M(m-r)}\}$. Using these representations in \eqref{eq.ZF.3} and block-partitioned representation of $\widetilde{\bW}_1$,
\bal
\widetilde{\bW}_1 =
\left(
  \begin{array}{cc}
    \bB_{11} & \bB_{12} \\
    \bB_{21} & \bB_{22} \\
  \end{array}
\right)
\eal
one obtains
\bal \notag
\label{eq.ZF.8}
\lambda\bI - \bLam_M &=
\left(
  \begin{array}{cc}
    \bB_{11} & \bB_{12} \\
    \bB_{21} & \bB_{22} \\
  \end{array}
\right)
\left(
  \begin{array}{cc}
    \bI_r-\lambda\bLam_r & 0 \\
    0 & \bI_{m-r} \\
  \end{array}
\right)+
\left(
  \begin{array}{cc}
    0 & 0 \\
    0 & \bA_{22} \\
  \end{array}
\right) \\
&=
\left(
  \begin{array}{cc}
    \bB_{11}(\bI_r-\lambda\bLam_r) & \bB_{12} \\
    \bB_{21}(\bI_r-\lambda\bLam_r) & \bB_{22}+\bA_{22} \\
  \end{array}
\right)
\eal
so that $\bB_{12}=\bB_{21}^{\dag}=0$ and $\bB_{11}>0$ is diagonal. This proves that the active eigenvectors of $\bR$ are also active eigenvectors of $\bW_1$ (note however that $\bW_1$ can have more active eigenvectors than $\bR$ but the converse is not true). No definite statements can be made at this point about inactive eigenvectors of $\bW_1$ and active eigenvectors of $\bW_2$, e.g. they do not have to be equal. The upper left block in \eqref{eq.ZF.8} implies \eqref{eq.ZF.2}.


\end{document}